\documentclass[twocolumn,pra,aps,showpacs,superscriptaddress,floatfix,showkeys,nofootinbib,10pt]{revtex4-2}

\usepackage[T1]{fontenc}
\usepackage{amsmath,amsfonts,amssymb,amsthm}
\usepackage{graphicx}
\usepackage{booktabs}
\usepackage{algorithm}
\usepackage{algpseudocode}

\newcommand{\ket}[1]{| #1 \rangle}
\newcommand{\bra}[1]{\langle #1 |}
\newcommand{\proj}[1]{\ket{#1}\bra{#1}}

\newcommand{\Ab}[1]{ \left| #1 \, \right|} 
\newcommand{\algeq}{\mathrel{\overset{\mathrm{alg}}{=}}}
\newcommand{\notalgeq}{\mathrel{\overset{\mathrm{alg}}{\neq}}}

\newcommand\bigforall{\mbox{\Large $\mathsurround=0pt\forall$}}

\newtheorem{definition}{Definition}
\newtheorem{lem}{Lemma}
\newtheorem{conj}{Conjecture}

\newtheorem{result}{Result}

\DeclareMathOperator{\Tr}{Tr}

\DeclareMathOperator{\Prob}{Prob}

\usepackage{color}

\begin{document}

\preprint{APS/123-QED}

\title{Picking NPA constraints from a randomly sampled quantum moment matrix}

\author{Giuseppe Viola}
\affiliation{Department of Physics, University of Siegen, Walter-Flex-Straße 3, 57068 Siegen, Germany }
\author{Anubhav Chaturvedi}
\affiliation{Faculty of Applied Physics and Mathematics, Gdańsk University of Technology, Gabriela Narutowicza 11/12, 80-233 Gdańsk, Poland}
\author{Piotr Mironowicz}
\affiliation{Department of Algorithms and System Modeling, Faculty of Electronics, Telecommunications and Informatics, Gda\'{n}sk University of Technology, Poland}
\affiliation{International Centre for Theory of Quantum Technologies, University of Gda\'{n}sk, Wita Stwosza 63, 80-308 Gda\'{n}sk, Poland}

\date{\today}

\begin{abstract}
	We introduce a simple and flexible method for constructing semidefinite programming relaxations of quantum correlations based on the Navascués--Pironio--Acín (NPA) hierarchy. Instead of deriving algebraic constraints explicitly, our approach identifies them numerically by sampling random moment matrices arising from quantum realizations.
	We show that, under broad conditions, a single random realization suffices to recover all constraints of the NPA moment matrix with probability $1$. We also characterize the regimes in which additional equalities may appear, relating them to the rank of measurement operators and the structure of the underlying scenario.
	Our method provides a practical alternative to standard algebraic constructions and enables efficient analysis of a wide range of operational scenarios, including cases with constraints on the rank of measurements.
\end{abstract}

\maketitle

\section{Introduction}
\label{sec:introduction}

Recently, a significant increase in interest in the use of quantum information processing methods for tasks traditionally associated with classical Information and Communication Technology (ICT), such as cryptography, has been observed. One of the most innovative approaches in this direction is the so-called device-independent (DI) cryptography~\cite{Yao98,DIreview}, which allows for the design of cryptographic protocols based on quantum technologies, such as quantum entanglement~\cite{H4}, providing security guarantees even in scenarios where the internal functioning of the devices, or their vendor, is not trusted.

The rapid development of DI cryptography would not be possible without the parallel development of numerical methods enabling its effective analysis, such as the Navascués--Pironio--Acín (NPA) hierarchy~\cite{NPA1,NPA2}. One of its notable successes is its application to experimental data analysis, leading to the first generation of data certified to be random under the laws of quantum physics~\cite{Pironio10}.

The NPA hierarchy provides a sequence of outer approximations to the set of quantum correlations via semi-definite programming (SDP)~\cite{doi:10.1137/1038003,mironowicz2024semi}. Many research groups employ private implementations of this method~\cite{PiotrThesis}, which often have limited functionality due to the complexity of constructing NPA relaxations for an arbitrary number of parties and hierarchy levels. The only widely used implementations are the \texttt{ncpol2sdpa} package~\cite{ncpol2sdpa} by Peter Wittek and a Julia implementation by Erik Woodhead~\cite{woodhead2024quantumnpa}.

The NPA method characterizes correlations arising from quantum mechanics via moment matrices (see Sec.~\ref{ssec:NPA}). Its construction can be understood as describing the closure of correlations achievable with finite-dimensional quantum systems. Similar techniques appear in optimization problems where additional constraints are imposed on the underlying probability distributions. One such framework is the semi-device-independent approach~\cite{Marcin}, where stronger conclusions can be obtained under physically motivated assumptions, such as bounds on the Hilbert space dimension. Several methods inspired by NPA have been developed to analyze such scenarios~\cite{dimWit1,dimWit2}.

These optimization techniques belong to the class of convex optimization problems, which are typically more tractable numerically. At the same time, there is ongoing work on incorporating additional constraints, such as bounds on the rank of operators~\cite{Otfried}. However, such constraints are inherently non-convex, which makes their incorporation into SDP-based frameworks challenging~\cite{rank1,rank2}.

In this work, we present a flexible and easily implementable technique that allows one to assess quantum behaviors across a broad class of operational scenarios where the NPA hierarchy applies. Instead of explicitly deriving algebraic relations between entries of the moment matrix, we identify these constraints numerically by generating random instances of quantum states and measurements, constructing the associated moment matrices, and recording which entries are equal (within a prescribed numerical tolerance) or vanish. Thus, algebraic relations are replaced by empirically determined numerical equalities. In addition to Bell scenarios, our technique can be applied to other operational settings, such as contextuality scenarios and prepare-and-measure scenarios with bounded communication capacity.

The structure of the paper is as follows. In Section~\ref{sec:Preliminaries} we introduce the basic notions related to optimization over quantum probabilities and recall the NPA hierarchy. In Section~\ref{sec:SDP} we briefly review the framework of semi-definite programming. Section~\ref{sec:Result} contains the main results of the paper, including the proposed moment matrix generation technique and its validation through numerical tests. Finally, in Section~\ref{sec:Conclusions} we summarize our findings and discuss their implications. Technical details, definitions, and proofs are deferred to the appendix.

Our results suggest that, for moment matrices associated with quantum measurements, random realizations preserve precisely the relations imposed by the underlying algebraic structure. In particular, additional equalities arise only in specific degenerate cases, such as rank-$1$ representations.

\section{Preliminaries}
\label{sec:Preliminaries}

\subsection{Optimization over quantum probabilities}
\label{ssec:quantum_probabilities}

A problem that is very profound for quantum information theory is the study of a setup described in terms of some set of parties, where each party possesses a measuring device with some number of possible settings and can produce an outcome from some specified set of allowed outcomes. In this way of analyzing the problem, the possible joint conditional probability distributions of outcomes obtained by individual parties, conditioned on the settings each party can input into their measuring device, are considered.

At this point, it is necessary to specify which types of devices are considered, as this will affect the range of joint probability distributions that can be obtained.

A certain class of probability distributions is the so-called local distributions with hidden variables. They are defined as the set of all distributions obtainable as follows. All the devices of the parties share some random variable. The device of each of the parties for a given setting provides a random outcome according to some probability distribution, depending on the setting and the shared random variable. One can derive that the set of probabilities that can be generated in this way is given by:
\begin{equation}
	\label{eq:LHV}
	P(a,b|x,y) = \sum_\lambda p(\lambda) P^{(A)}(a|x,\lambda) P^{(B)}(b|y,\lambda),
\end{equation}
where $P^{(A)}(a|x,\lambda)$, $P^{(B)}(b|y,\lambda)$ and $p(\lambda)$ are probability distributions of Alice, Bob, and the shared variable, respectively.

Another class is the so-called non-signaling probability distributions. These probabilities are restricted to satisfy the non-signaling condition stating that:
\begin{equation}
	\label{eq:NS}
	P(a|x) = P(a|x,y) \equiv \sum_b P(a,b|x,y), \ \bigforall a,x,y.
\end{equation}
This condition expresses the fact that, if the devices are separated in space, then their operation cannot allow for instantaneous signaling, which would be the case if the local probabilities of one party depended on the setting of another party.

Quantum distributions are defined as all probability distributions obtainable by quantum measurements performed by each of the devices, each of which has access to some sub-system of a shared quantum state. The assumption that each device only acts on its own sub-system ensures that the non-signaling condition~\eqref{eq:NS} is satisfied. Algebraically, this is expressed by the fact that measurement operators of different parties commute. On the other hand, operators acting on the same sub-space need not commute, meaning that the result of their action may depend on the order in which they are applied.

When analyzing quantum distributions, various special cases are considered, depending on the dimensions of the sub-spaces of each party. In general, when one speaks about quantum probability distributions without assumptions about their physical realization (often called boxes or behaviors), it is assumed that the underlying quantum systems can be of arbitrary, possibly infinite dimension. In this nomenclature, one speaks about local (or classical), non-signaling, or quantum boxes, respectively.

One of the key results for quantum information theory was the observation that some joint probability distributions obtainable in quantum mechanics cannot be reproduced using local boxes~\cite{EPR35,Bell,CHSH,Aspect}. It was also observed that, to determine whether a given box is contained in one of the mentioned classes, one needs to perform different types of calculations, often involving optimization over a set of boxes. In that case, one assigns a variable to each event and constructs a target function as a linear combination of such variables.

In the case of non-signaling boxes, this is straightforward, as it suffices to verify whether a given box satisfies the conditions~\eqref{eq:NS}. Consequently, the problem of optimization over non-signaling boxes can be solved using linear programming (LP)~\cite{murty1983linear}, where the non-signaling conditions are imposed as constraints.

In the case of local boxes, one can observe that any such joint probability distribution is a convex combination of deterministic boxes, i.e., boxes where the outcomes for each setting of each party are pre-determined~\cite{LHV}. Thus, the problem of optimization over local boxes is also an LP, which in practice requires enumeration of all deterministic boxes. The number of such boxes grows rapidly with the parameters of the setup.

Simultaneously, it can be shown that the problem of optimization over quantum boxes is NP-hard~\cite{Gurvits,DPS,Kempe}. Fortunately, even though one cannot optimize exactly over this set, relaxations can be defined using the NPA hierarchy~\cite{NPA1,NPA2}, which are tractable via SDP.

\subsection{Navascues-Pironio-Acin hierarchy}
\label{ssec:NPA}

The core idea in defining the NPA hierarchies is the notion of the so-called moment matrices~\cite{lasserre2010moments}. Their concept was introduced earlier in probability theory, where a moment matrix is defined as a square symmetric matrix containing moments of a probability distribution. In moment matrices the rows and columns are indexed by monomials.

In the NPA method, the monomials involve non-commuting variables, with each variable related to a quantum measurement operator. To keep the notation simple, we restrict our considerations to bi-partite probability distributions, with two parties Alice and Bob. Let us denote by $A^a_x$ the projector related to the outcome $a$ of Alice for the setting $x$, and similarly for Bob $B^b_y$ for his setting $y$ and outcome $b$. Let $\ket{\psi}$ be the quantum state they access. Then, according to the formalism of quantum theory, their joint probability of obtaining the outcomes $a$ and $b$ for the settings $x$ and $y$ is given by
\begin{equation}
	\label{eq:Pabxy}
	P(a,b|x,y) = \bra{\psi} A^a_x B^b_y \ket{\psi}.
\end{equation}
Let $N^A(x)$ be the number of outcomes of Alice for the setting $x$, and $N^B(y)$ for Bob. The definition of quantum measurements requires
\begin{equation}
	\label{eq:sumIdentity}
	\bigforall x, \sum_{a = 1}^{N^A(x)} A^a_x = \openone,
\end{equation}
and the same for Bob.

For any quantum probability distribution $\{P(a,b|x,y)\}_{a,b,x,y}$, any set of quantum state and measurements which reproduce it via~\eqref{eq:Pabxy} is called a \textit{quantum realization} of the distribution.

Consider a sequence of operators $S = O_1 O_2 \cdots O_k$, where $O_i \in \mathcal{E} \equiv \{A^a_x\}_{a,x} \cup \{B^b_y\}_{b,y}$. For a given set of sequences $\mathcal{S} = \{S_k\}_{k}$, let us define a moment matrix $\Gamma$ indexed by $\mathcal{S}$ related to $\ket{\psi}$, $\{A^a_x\}_{a,x}$ and $\{B^b_y\}_{b,y}$ as:
\begin{equation}
\label{Gamma_entry}
	\Gamma_{k_1, k_2} \equiv \bra{\psi} S_{k_1}^{\dagger} S_{k_2} \ket{\psi}.
\end{equation}
It follows that
\begin{equation}
	\label{eq:Gammak1k2}
	\Gamma_{k_1, k_2} = \Gamma_{l_1, l_2}
\end{equation}
whenever
\begin{equation}
	\label{eq:GammaEqs}
	S_{k_1}^{\dagger} S_{k_2} = S_{l_1}^{\dagger} S_{l_2},
\end{equation}
i.e., whenever the corresponding operator products are equal.

Now, consider the set $\mathcal{S}_n$ defined as the set of all sequences of operators from $\mathcal{E}$ of length at most $n$, where the sequence of length $0$ is $\openone$. A moment matrix of level $n$ is the $\Gamma$ matrix indexed by $\mathcal{S}_n$.

One often defines intermediate levels for sequences, e.g.
\begin{equation}
	\mathcal{S}_{1+AB} \equiv \mathcal{S}_1 \cup \left\{ A^a_x B^b_y \right\}_{x,y, a \neq N^A(x), b \neq N^B(y)}.
\end{equation}
One can observe that because of~\eqref{eq:sumIdentity} the operators in $\mathcal{S}_n$ are not linearly independent. Consider a reduced set of operator sequences
\begin{equation}
	\label{eq:reducedOperators}
	\hat{\mathcal{E}} \equiv \left( \{A^a_x\}_{a,x} \cup \{B^b_y\}_{b,y} \right) \setminus \left\{ A^{N^A(x)}_x, B^{N^B(y)}_y \right\}_{x,y},
\end{equation}
where the operators in the latter set can be obtained as linear combinations of the former using~\eqref{eq:sumIdentity}.
Moment matrices built from the reduced set of operators $\hat{\mathcal{E}}$ will be referred to as the minimal form of the moment matrix, or simply reduced moment matrix. The reduced moment matrix is constructed analogously, with monomials involving projectors corresponding to the last outcome for each setting removed. This does not change the properties of the moment matrix, since the removed rows and columns are determined by the remaining ones. The minimal form requires fewer resources in numerical applications and typically has a better condition number.

A direct property of the moment matrix $\Gamma$ is its positive semi-definiteness. This follows from the fact that each entry $\Gamma_{k_1, k_2}$ is an inner product of vectors $S_{k_1} \ket{\psi}$ and $S_{k_2} \ket{\psi}$~\cite{Boyd04}.
Furthermore, from normalization, the element corresponding to $\bra{\psi} \openone \ket{\psi}$ is equal to $1$.
By construction, some elements of the moment matrix must be equal, which is implied by commutation and orthogonality properties of the projectors.

The commutation property is given by $\bigforall a,b,x,y$,
\begin{equation}
	 [A^a_x,B^{b}_{y}] = 0,
\end{equation}
while orthogonality is captured by $\bigforall x,a,a'$ such that $a\neq a'$,
\begin{equation}
	A^a_x A^{a'}_{x} = 0,
\end{equation}
and similarly for Bob.

So, NPA allows one to transform optimization problems of the form
\begin{equation}
	\max_{\{P(a,b|x,y)\}_{a,b,x,y} \in \mathbf{Q}} \sum_{a,b,x,y} c_{a,b,x,y} P(a,b|x,y),
\end{equation}
with $\mathbf{Q}$ being the set of quantum probability distributions and $c_{a,b,x,y} \in \mathbf{R}$, into relaxations of the form
\begin{equation}
	\max_{\{P(a,b|x,y)\}_{a,b,x,y} \in \mathbf{Q_n}} \sum_{a,b,x,y} c_{a,b,x,y} P(a,b|x,y),
\end{equation}
where $\mathbf{Q_n}$ denotes the set of behaviors compatible with a moment matrix $\Gamma^{(n)}$ of level $n$ satisfying
\begin{equation}
	\Gamma^{(n)} \geq 0,
\end{equation}
and
\begin{equation}
	\Gamma^{(n)}_{0,0} = 1.
\end{equation}

Note that these problems reduce to membership problems when $c_{a,b,x,y} = 0, \ \bigforall a,b,x,y$.

\section{Semi-definite programming}
\label{sec:SDP}

One of the breakthrough moments in the development of modern methods of numerical optimization was the formulation and development of linear programming (LP). LP is an optimization method in which a number of variables are considered, subject to linear constraints, and the quantity to be optimized, called the objective function, is linear in these variables. The first formulation of this technique was given in 1939 by Leonid Kantorovich~\cite{kantorovich1940method}.

The first numerical method to solve such optimization problems, the so-called simplex method, was given in 1947 by George B. Dantzig~\cite{dantzig1951maximization}. Soon after, Dantzig and John von Neumann developed the foundations of the theory of duality~\cite{10.2307/2296111}, pointing out that for any LP one can construct a second problem, the so-called dual, in which the elements of the initial problem (called primal) play different roles, while both problems share the same optimal value. LP reached a mature form in 1984 after Narendra Karmarkar introduced the interior point method~\cite{Karmarkar}, allowing efficient solution of such problems.

Semi-definite programming (SDP) is an optimization technique that generalizes LP. The primal problem has the form:
\begin{equation}
\begin{aligned}
	\text{minimize } & \Tr(C X) \\
	\text{subject to } & \Tr(A_i X) = b_i, \ \text{for } i = 1, \cdots, m, \\
	& X \succeq 0,
\end{aligned}
\end{equation}
The dual problem has the form:
\begin{equation}
\begin{aligned}
	\text{maximize } & b^{T} y \\
	\text{subject to } & C - \sum_{i=1}^{m} y_i A_i = Z, \\
	& Z \succeq 0.
\end{aligned}
\end{equation}
In the above, symmetric matrices $C \in \mathbb{S}^n$ and $A_1, \cdots, A_m \in \mathbb{S}^n$, and vector $b \in \mathbb{R}^m$, specify the SDP problem.

\section{Results}
\label{sec:Result}

The following result relies on Conjecture~\ref{conjecture1}, which is strongly supported by numerical evidence (see Appendix~\ref{app:randomization_results}).

\subsection{Moment matrix generation technique}
\label{ssec:technique}

Our technique consists in generating the list of constraints satisfied by the moment matrix associated with a given problem by randomly sampling a single moment matrix and recording which entries are equal to each other and which entries are equal to zero.

The random generation of the moment matrix associated with a problem is performed by randomly generating a quantum state shared by the parties involved in the communication task, and a set of projectors for each measurement setting of each of the agents. When randomly sampling projectors, in order to avoid constraints related to the rank of the generated projectors, it is sufficient to consider projectors of rank at least $2$. However, as shown in Appendix~\ref{app:appendixLemma}, in some cases even projectors of rank $1$ do not introduce additional constraints associated with their rank.

In Appendix~\ref{app:appendixLemma} we derive a lemma stating the conditions under which two entries of a moment matrix are equal, in the case when the associated projectors and quantum state are randomly generated. As a consequence, we obtain the following result:

\begin{result}
\label{Res1}
	For any NPA problem, the constraints on the moment matrix that follow from its definition~\eqref{eq:Gammak1k2}, i.e.\ from the fact that each entry has the form~\eqref{Gamma_entry}, can be identified from a moment matrix associated with a randomly generated quantum realization with probability $1$, provided that at least one of the following conditions is satisfied:

    \begin{enumerate}
	\item Calling $L$ the level of the NPA method,
	\begin{equation}
		\label{lev3}
		L < 3.
	\end{equation}
    \item 
	\begin{equation}
		r > 1,
	\end{equation}
    where $r$ is the rank of the measurements in the randomly generated quantum realization.
    \item No agent receives either:
	\begin{itemize}
		\item at least $3$ inputs each associated with at least $2$ outputs, or
		\item $2$ inputs associated with at least $2$ and $3$ outputs, respectively.
	\end{itemize}
    \end{enumerate}
\end{result}

\begin{proof}
	In Lemma~\ref{conj:probability_equality} presented in the appendix, in which it is assumed Conjecture~\ref{conjecture1}, we prove that, when $r>1$, picking equalities among the entries of a randomly generated moment matrix does not, with probability $1$, produce additional constraints beyond those obtainable algebraically using only properties of projectors, as described in Lemma~\ref{lem:moment_entry_equality}. On the other hand, randomly generated projectors satisfy all algebraic relations by construction.
	
	When $r=1$, Lemma~\ref{conj:probability_equality} implies that additional equalities may emerge from the associated moment matrix. In particular, Lemma~\ref{short_homogeneous}, together with Lemma~\ref{conj:probability_equality}, implies that, in order for such additional equalities to appear, products of at least $5$ projectors associated with the same party must occur within a single entry of the moment matrix. This requires at least $L=3$.
	
	Furthermore, Lemma~\ref{smallest}, together with Lemma~\ref{conj:probability_equality}, implies that these additional equalities arise only when an agent receives either:
	\begin{itemize}
		\item at least $3$ inputs, each associated with at least $2$ outputs, or
		\item $2$ inputs associated with at least $2$ and $3$ outputs, respectively.
	\end{itemize}
\end{proof}

\subsection{Testing the Tool}
\label{sec:testing_tool}

For any Bell scenario and any level of the NPA hierarchy, there exists a characteristic set of equalities implied by the NPA constraints. This set, determined by the algebra of the problem, can be obtained as an output of \texttt{ncpol2sdpa}. To showcase the efficacy of our method, we compared this characteristic set of equalities with the set of equalities extracted from a randomly generated moment matrix for a broad range of Bell scenarios and levels of the NPA hierarchy (see Table~\ref{tabRes}). 

\begin{table}[htbp]
  \centering
  \caption{Distinct entries in the level $3$ NPA moment matrix for Bell scenarios $(X,Y,A,B)$. Counts from algebraic constraints (identical to those from randomly sampled rank-$2$ projectors) are compared with counts from rank-$1$ projectors. Rank-$1$ sampling induces extra equalities (fewer unique entries) in every case except CHSH $(2,2,2,2)$, where both counts coincide.}
  
  \label{tab:equality-counts}
  \begin{tabular}{|c|c|c|c|p{2.5cm}|p{2.5cm}|}
    \hline
    $X$ & $Y$ & $A$ & $B$ & \# unique entries (algebraic \& rank-2) & \# unique entries (rank-1) \\
    \hline
    2 & 2 & 2 & 2 & 61   & 61   \\
    2 & 2 & 2 & 3 & 422  & 410  \\
    2 & 2 & 3 & 3 & 1449 & 1412 \\
    2 & 3 & 2 & 2 & 319  & 292  \\
    2 & 3 & 2 & 3 & 7048 & 6495 \\
    2 & 3 & 3 & 2 & 1122 & 1077 \\
    2 & 3 & 3 & 3 & 12531 & 11919 \\
    3 & 3 & 2 & 2 & 868  & 808  \\
    3 & 3 & 3 & 2 & 10438 & 9822 \\
    3 & 3 & 3 & 3 & 38017 & 36717 \\
    \hline
  \end{tabular}
  \label{tabRes}
\end{table}

In all cases reported in Table~\ref{tabRes}, when the set of equalities is extracted from a randomly generated moment matrix with rank-$2$ projectors, the sets exactly match, demonstrating the robustness of our method.

We also compared the set of equalities extracted from a randomly generated moment matrix with rank-$1$ projectors. We found that, for level $3$ of the NPA hierarchy, for all Bell scenarios for which condition $3$ of Result~\ref{Res1} is violated, additional equalities arise compared to those obtained from \texttt{ncpol2sdpa}. One prominent example of this effect is level $3$ of the NPA hierarchy for the $(3,3,2,2)$ Bell scenario.

\section{Conclusions}
\label{sec:Conclusions}

In this work, we have presented an alternative technique for identifying the constraints implied by the NPA hierarchy. The proposed approach is simple to implement, computationally efficient, and easily adaptable to a wide range of scenarios.

The main idea is to replace the explicit derivation of algebraic relations with their numerical identification via random sampling of quantum realizations. We have shown that, under broad conditions, this procedure recovers exactly the constraints imposed by the NPA hierarchy with probability $1$. This provides a practical and flexible tool for constructing moment matrices without relying on a detailed algebraic analysis.

We have also analyzed the limitations of the method. In particular, we identified the precise conditions under which additional, non-algebraic equalities may appear. These arise in the case of rank-$1$ projectors and are directly linked to the structure of homogeneous blocks of projectors. Our results characterize the minimal scenarios in which such effects occur.

Finally, the presented framework allows for a more refined analysis of scenarios involving constraints on the rank of measurement operators. In particular, it enables a systematic study of rank-$1$ measurements, which are relevant in various semi-device-independent settings. Overall, the proposed method offers a conceptually simple and practically effective alternative to standard algebraic constructions in the study of quantum correlations.

Beyond the specific application to the NPA hierarchy, our results indicate that random realizations of quantum states and measurements faithfully reflect the algebraic structure of the associated moment matrices. In particular, equalities between entries are, with high probability, exactly those dictated by the algebra of the underlying operators, with additional equalities appearing only in well-characterized degenerate cases such as rank-$1$ representations. This perspective suggests that random sampling provides a reliable and practical tool for identifying the intrinsic structure of moment matrix constraints. An interesting open question is whether the observed correspondence between algebraic and numerically identified constraints can be established rigorously in full generality, and to what extent similar phenomena occur for other classes of operator relations and hierarchies beyond the NPA framework.

\begin{acknowledgments}
GV acknowledges support from the Ministry of Culture and Science of North Rhine-Westphalia via the NRW Rückkehrprogramm.
\end{acknowledgments}

\newpage
\onecolumngrid

\appendix

In this appendix, we provide the technical foundations underlying the results presented in the main text. We begin in Appendix~\ref{app:random_state} by describing the procedures used to generate random quantum states and measurements. In Appendix~\ref{app:notation} we introduce the notation and formal definitions used throughout the analysis. Appendix~\ref{app:algebraic} characterizes algebraic equalities between entries of the moment matrix. In Appendix~\ref{app:random_generation} we formalize the random generation model and express moment matrix entries as functions of the underlying random parameters. In Appendix~\ref{app:randomization_results} we present numerical evidence together with intermediate results on the behavior of moment matrix entries under random sampling. Finally, in Appendix~\ref{app:appendixLemma} we state the main probabilistic conjecture underlying the results used in the paper.

\section{Random generation of quantum states and measurements}
\label{app:random_state}

\begin{algorithm}[H]
	\begin{algorithmic}[1]
		\Require Dimension $d \in \mathbb{N}$
		\Ensure A random density matrix $\rho \in \mathbb{C}^{d \times d}$
		\State Sample two $d \times d$ matrices $R_1, R_2$ with i.i.d. standard normal entries
		\State Form the complex matrix $M \gets R_1 + i R_2$
		\State Compute $\Gamma \gets M M^\dagger$ \Comment{Ensures positive semi-definiteness}
		\State Normalize: $\rho \gets \Gamma / \Tr(\Gamma)$
		\State \Return $\rho$
	\end{algorithmic}
	\caption{Sampling Haar-Random Density Matrix~\cite{Zyczkowski2011}}
	\label{alg:HaarDensity}
\end{algorithm}

A key component in benchmarking and testing SDP-based quantum tools is the ability to generate random quantum states and measurements. In our implementation, we generate random density matrices using a unitarily invariant construction based on the Haar measure. Algorithm~\ref{alg:HaarDensity} describes the procedure used to generate random quantum states of dimension $d$. The randomization procedure of $R_1$ and $R_2$ can be implemented using the function \texttt{randn(d,d)} available natively in MATLAB or in the \texttt{numpy} library in Python.

This method produces density matrices distributed according to the induced (Hilbert--Schmidt) measure, which is unitarily invariant. This approach is particularly useful for generating random instances in simulation-based validation of quantum information protocols and SDP relaxations such as those used in the NPA hierarchy.

\subsection*{Random Generation of Rank-$r$ Projectors in Dimension $d$}
\label{app:random_proj}

We describe a generic procedure to generate a set of orthogonal projectors of rank-$r$ in a Hilbert space $\mathcal{H}$ of dimension $d$, assuming that $d$ is a multiple of $r$. This construction produces a collection of $n = \frac{d}{r}$ projectors $\{ P_j \}_{j=0}^{n-1}$ satisfying:
\begin{equation}
	P_j^2 = P_j, 
	\quad 
	P_j P_k = 0 \ \text{for } j \neq k,
	\quad 
	\sum_{j=0}^{n-1} P_j = \openone,
\end{equation}
where each projector has rank $r$. 

\begin{algorithm}[H]
	\begin{algorithmic}[1]
		\Require Number of outcomes $n \in \mathbb{N}$, desired rank $r \in \mathbb{N}$.
		\Ensure A set $\{P_j\}_{j=0}^{n-1}$ of $n$ mutually orthogonal projectors on $\mathbb{C}^{d}$, each of rank $r$, where $d = nr$, satisfying $P_j P_k = \delta_{jk} P_j$ and $\sum_{j=0}^{n-1} P_j = \openone$.
		\State Set the Hilbert space dimension $d \gets n r$.
		\State Sample a Haar-random unitary $U \in \mathbb{C}^{d\times d}$ using Algorithm~\ref{alg:haar_unitary_qr}.
		\For{$j = 0,1,\dots,n-1$}
		  \State Define the rank-$r$ projector onto the $j$-th block of $r$ columns of $U$:
		  \[
		    P_j \gets \sum_{i=1}^{r} \proj{u_{i+jr}},
		  \]
		  \Statex \hspace{1.7em}where $\{ \ket{u_\ell} \}_{\ell=1}^{d}$ are the orthonormal columns of $U$.
		\EndFor
		\State \Return $\{P_j\}_{j=0}^{n-1}$.
	\end{algorithmic}
	\caption{Haar-random rank-$r$ projective measurement with $n$ outcomes~\cite{ZyczkowskiSommers2001}}
	\label{alg:haar_random_rank_r_projective_measurement}
\end{algorithm}

\begin{algorithm}[H]
	\begin{algorithmic}[1]
		\Require Dimension $d \in \mathbb{N}$.
		\Ensure A Haar-random unitary $U \in \mathbb{C}^{d\times d}$.
		\State Draw a complex Ginibre matrix $G \in \mathbb{C}^{d\times d}$ with i.i.d.\ entries
		\[
		G_{ab} \sim \mathcal{N}(0,1) + i\,\mathcal{N}(0,1).
		\]
		\Statex \hspace{1.7em}\textit{Here, $\mathcal{N}(0,1)$ denotes the standard real normal distribution,}
		\Statex \hspace{1.7em}\textit{and the real and imaginary parts are sampled independently.}
		\State Compute the QR decomposition $G = QR$, with $Q$ unitary and $R$ upper triangular.
		\State Let $\Lambda \gets \mathrm{diag}(R)$ and define
		\[
		\phi_\ell \gets \frac{\Lambda_\ell}{|\Lambda_\ell|} \quad (\text{set } \phi_\ell \gets 1 \text{ if } \Lambda_\ell = 0), \quad \ell=1,\dots,d.
		\]
		\State Set $D \gets \mathrm{diag}(\overline{\phi_1},\dots,\overline{\phi_d})$ and output
		\[
		U \gets QD.
		\]
		\State \Return $U$.
	\end{algorithmic}
	\caption{Haar-random unitary via QR decomposition~\cite{Mezzadri2007}}
	\label{alg:haar_unitary_qr}
\end{algorithm}

This algorithm yields a set of $n = \frac{d}{r}$ orthogonal projectors of rank $r$ forming a resolution of the identity.

So, given the number of outputs $n$ associated with a specific input defined by the considered scenario, the above construction implies a minimum dimension $d = n r$ for the generated projectors. It is also possible to generate projectors in higher dimensions $d' > d$. In that case, one of the projectors, e.g.\ the last one, can be defined as the identity minus the sum of the remaining ones, and thus will have rank $r' > r$. This does not pose a problem for our purposes, since we consider moment matrices in which only projectors associated with $n-1$ outcomes appear. Consequently, all projectors explicitly present in the moment matrix still have rank $r$.

\begin{algorithm}[H]
	\begin{algorithmic}[1]
		\Require Hilbert space dimension $d$, desired rank $r$ with $r \mid d$.
		\Ensure A set of $n = d/r$ orthogonal projectors $\{P_j\}_{j=0}^{n-1}$ of rank $r$.
		\State Generate a random density matrix $\rho \in \mathbb{C}^{d \times d}$ using Algorithm~\ref{alg:HaarDensity}.
		\State Compute the eigenvectors $\{ \ket{v_i} \}_{i=1}^{d}$ of $\rho$ (they form an orthonormal basis of $\mathcal{H}$).
		\For{$j = 0, 1, \dots, n-1$}
		\State Construct the projector
		\begin{equation}
			P_j := \sum_{i=1}^{r} \proj{v_{i + jr}}.
		\end{equation}
		\EndFor
		\State \Return $\{P_j\}_{j=0}^{n-1}$.
	\end{algorithmic}
	\caption{Random generation of rank-$r$ orthogonal projectors in dimension $d$~\cite{NielsenChuang2010}}
	\label{alg:random_rank_r_projectors_2}
\end{algorithm}

Alternatively, one can generate sets of projectors using Algorithm~\ref{alg:random_rank_r_projectors_2}, which relies on the fact that a random density matrix generated as in Algorithm~\ref{alg:HaarDensity} is almost surely of full rank.

\section{Notation and definitions}
\label{app:notation}

We adopt the following notation and conventions throughout the paper.

\begin{definition}[Scenario sequence and scenario operator]
	\label{dfn:seq_of_projs}
	Let $m, r \in \mathbb{N}_{+}$ and let $\mathcal{H} = \bigotimes_{i=1}^{m} \mathcal{H}^{(i)}$ be a composite Hilbert space, where each $\mathcal{H}^{(i)}$ corresponds to the $i$-th party. Let $d = \dim(\mathcal{H})$.
	
	Let 
	\begin{equation}
		\mathbb{P} = \left( (P^{(i)}_{j})_{j=1}^{J_i} \right)_{i=1}^{m}
	\end{equation} 
	with all $J_i \in \mathbb{N}_{+}$, be a finite sequence of sequences of projectors acting on $\mathcal{H}$. We say that $\mathbb{P}$ is a \emph{scenario sequence} if the following conditions hold:
	\begin{enumerate}
		\item \textbf{Party assignment:} for each $i,j$, the operator $P^{(i)}_j$ acts non-trivially only on $\mathcal{H}^{(i)}$ (i.e.\ as $P^{(i)}_j \otimes \openone$ on the remaining subsystems);
		\item \textbf{Fixed rank:} each projector $P^{(i)}_j$ has rank exactly $r$ in the space on which it acts;
		\item \textbf{No last outcome:} the projectors $\{P^{(i)}_j\}_j$ do not include the projector associated with the last outcome, since it can be derived from the others through the completeness relation.
	\end{enumerate}
	
	We define the associated \emph{scenario operator} as
	\begin{equation}
		\mathbb{\hat{P}} = \left( \prod_{j=1}^{J_1} P^{(1)}_{j} \right) \otimes \cdots \otimes \left( \prod_{j=1}^{J_m} P^{(m)}_{j} \right),
	\end{equation} 
	acting on the Hilbert space $\mathcal{H} = \mathcal{H}^{(1)} \otimes \cdots \otimes \mathcal{H}^{(m)}$.
\end{definition}

In the context of randomly generated quantum states and measurements, the quantities appearing in Definition~\ref{dfn:seq_of_projs} are interpreted as follows: $r$ denotes the rank of each generated projector, $d$ is the dimension of the quantum state $\rho$, and $m$ is the number of parties involved in the scenario.

The moment matrices considered in this work are constructed by omitting the projectors corresponding to the final outcomes of each measurement, leveraging the fact that the full set of projectors for a given measurement sums to the identity.

\begin{definition}[Blocks of projectors]
	\label{dfn:block_of_projs}
	Let $\mathbb{P}$ be a scenario sequence. For each party $i$, define $\mathbb{P}^{(i)}$ as the subsequence of $\mathbb{P}$ consisting of all projectors assigned to party $i$, equipped with the order induced from $\mathbb{P}$. 
	
	We say that $\mathbb{P}^{(i)}$ is \emph{simplified} if for every pair of adjacent projectors in the sequence:
	\begin{enumerate}
		\item they are non-orthogonal, i.e.\ $P^{(i)}_j P^{(i)}_{j+1} \neq 0$;
		\item they are non-identical, i.e.\ $P^{(i)}_j \neq P^{(i)}_{j+1}$.
	\end{enumerate}
	
	If $\mathbb{P}^{(i)}$ is simplified, then we call it the \emph{$i$-th block of projectors} of the scenario sequence $\mathbb{P}$. We denote by $P^{(i)}_j$ the $j$-th projector in this block.
\end{definition}

Note that $P^{(i)}_j$ refers to the $j$-th element of the $i$-th block, and not to the $j$-th projector in the full scenario sequence $\mathbb{P}$. The condition of simplification implies that the operator associated with the block,
\begin{equation}
\label{eq:scenario_operator}
	\hat{\mathbb{P}}^{(i)} := P^{(i)}_1 P^{(i)}_2 \cdots P^{(i)}_{J_i},
\end{equation}
cannot be represented by a strictly shorter subsequence of projectors from $\mathbb{P}$ that yields the same operator. In other words, the block is minimal with respect to algebraic simplifications induced by orthogonality and idempotency.

\begin{definition}[Norm of a block of projectors]
Let $\mathbb{P}^{(i)}$ be the $i$-th simplified block of projectors. We define its \emph{norm} (or length) as
\begin{equation}
	\Ab{\mathbb{P}^{(i)}} := 
	\begin{cases}
		0 & \text{if } \hat{\mathbb{P}}^{(i)} = \openone, \\
		J_i & \text{otherwise}.
	\end{cases}
\end{equation}
\end{definition}

This norm captures the minimal number of projectors needed to represent the action of a block.

\begin{definition}[Multiset of consecutive pairs]
	Let $\mathbb{P} := (P_0, \dots, P_{n-1})$ be an ordered sequence of projectors. We define the \emph{multiset of consecutive pairs} of $\mathbb{P}$, denoted $\mathcal{C}(\mathbb{P})$, as the multiset (with multiplicities)
	\begin{equation}
		\mathcal{C}(\mathbb{P}) := \{ (P_i, P_{i+1}) \mid i \in \{0, \dots, n-2\} \}.
	\end{equation}
\end{definition}

\begin{definition}[Homogeneous blocks]
	\label{dfn:homogeneous_blocks}
	Let $\mathbb{A}^{(i)} = (A_0, \dots, A_{n-1})$ and $\mathbb{B}^{(i)} = (B_0, \dots, B_{n-1})$ be two simplified blocks of projectors associated with the $i$-th party. We say that $\mathbb{A}^{(i)}$ and $\mathbb{B}^{(i)}$ are \emph{homogeneous} if the following conditions hold:
	\begin{enumerate}
		\item $\Ab{\mathbb{A}^{(i)}} = \Ab{\mathbb{B}^{(i)}}$;
		\item $A_0 = B_0$ and $A_{n-1} = B_{n-1}$;
		\item $\mathcal{C}(\mathbb{A}^{(i)}) = \mathcal{C}(\mathbb{B}^{(i)})$ as multisets.
	\end{enumerate}
\end{definition}

The following example illustrates the notion of homogeneous blocks.

\textbf{Example 1.} Consider the following two simplified blocks of projectors (we omit the party index for simplicity), both of length $5$:
\begin{equation}
	\mathbb{P} = (P_0, P_1, P_0, P_2, P_0),
\end{equation}
\begin{equation}
	\mathbb{P}' = (P_0, P_2, P_0, P_1, P_0).
\end{equation}
The multiset of consecutive pairs for $\mathbb{P}$ is
\begin{equation}
	\mathcal{C}(\mathbb{P}) = \{ (P_0, P_1), (P_1, P_0), (P_0, P_2), (P_2, P_0) \},
\end{equation}
and for $\mathbb{P}'$:
\begin{equation}
	\mathcal{C}(\mathbb{P}') = \{ (P_0, P_2), (P_2, P_0), (P_0, P_1), (P_1, P_0) \}.
\end{equation}
We observe that both blocks have the same length, the same first and last projectors, and the same multiset of consecutive pairs. Therefore, by Definition~\ref{dfn:homogeneous_blocks}, the blocks $\mathbb{P}$ and $\mathbb{P}'$ are homogeneous.

\begin{lem}
	\label{short_homogeneous}
    The shortest length of two homogeneous simplified blocks of projectors, not identical as sequences, is $5$, and at least $3$ different projectors are required.
\end{lem}

\begin{proof}
By definition, simplified blocks of projectors cannot contain consecutive equal elements, and the first and last corresponding elements of the blocks must be equal. Additionally, by hypothesis, the two simplified blocks must not be identical as sequences. Therefore, it is not possible to construct such blocks using fewer than $3$ distinct projectors.

Since at least $3$ different projectors are required, the minimum length of the blocks cannot be smaller than $3$. However, for length $3$, the conditions of having identical first and last elements and identical multisets of consecutive pairs imply that the two blocks must be identical as sequences, hence this case is excluded.

Let us consider blocks of length $4$. The number of distinct projectors cannot exceed $4$, and the case of $4$ distinct projectors leads to identical sequences due to the constraints on the first and last elements and the multiset of consecutive pairs. Thus, the only remaining possibility is to use $3$ distinct projectors. By enumerating all such sequences, one finds that no pair satisfies simultaneously the conditions of being simplified, homogeneous, and non-identical.

For length $5$, consider the simplified blocks
\begin{equation}
\label{P00}
	\mathbb{P} = (P_0, P_1, P_0, P_2, P_0),
\end{equation}
\begin{equation}
\label{P11}
	\mathbb{P}' = (P_0, P_2, P_0, P_1, P_0).
\end{equation}
These blocks are homogeneous and not identical as sequences, which concludes the proof.
\end{proof}

Next, we have the following:
\begin{lem}
	\label{smallest}
	The smallest scenarios in which homogeneous simplified blocks, not identical as sequences, appear in the moment matrix are the following two:
	\begin{enumerate}
	    \item a single agent receives $2$ inputs to which are associated $2$ and $3$ outcomes, respectively,
	    \item a single agent receives $3$ inputs to which are associated $2$ outcomes each.
	\end{enumerate}
\end{lem}

\begin{proof}
	First, we show that more than one measurement setting is required. By Lemma~\ref{short_homogeneous}, we need at least $3$ distinct projectors. In the single-agent case, if there is only one input with at least $4$ outcomes (so that at least $3$ different projectors appear in the moment matrix), then all simplified blocks consist either of a single projector or of the identity. Therefore, at least $2$ inputs are necessary.

	The next smallest scenario with at least $3$ projectors is the one in which the agent has $2$ inputs with $2$ and $3$ outcomes, respectively. In this case, referring to~\eqref{P00} and~\eqref{P11}, let $P_0$ denote a projector associated with the first input, and let $P_1$ and $P_2$ denote two projectors associated with the second input. Then one obtains two homogeneous simplified blocks.

	The same situation occurs for $3$ inputs with $2$ outcomes each, in which case there are again $3$ distinct projectors, which can be denoted by $P_0$, $P_1$, and $P_2$.
\end{proof}

We distinguish between equality of numerical values ($=$), algebraic equality ($\algeq$), and equality of operator representations ($\equiv$).

\begin{definition}[Algebraic equality]
\label{dfn:algebraic_equality}
Let $\mathcal{E} = \{P^{(i)}_j\}_{i,j}$ be a set of projectors satisfying the relations:
\begin{equation}
	P^{(i)}_j P^{(i)}_k = 0 \ \text{for } j \neq k, 
	\quad 
	(P^{(i)}_j)^2 = P^{(i)}_j,
\end{equation}
and
\begin{equation}
	[P^{(i)}_j, P^{(i')}_{k}] = 0 \quad \text{for } i \neq i'.
\end{equation}

We say that two operator expressions $X$ and $Y$ built from $\mathcal{E}$ are \emph{algebraically equal}, denoted $X \algeq Y$, if $X$ can be transformed into $Y$ using only the above relations.
\end{definition}

\begin{definition}[Representational equality]
	\label{dfn:representational_equality}
	Let $X$ and $Y$ be expressions constructed from projectors and quantum states as in Section~\ref{app:random_generation}. We say that $X$ and $Y$ are \emph{representationally equal}, and write
	\begin{equation}
		X \equiv Y,
	\end{equation}
	if, after substituting the parametrization defined by the random variable $\omega$ (see Eq.~\eqref{eq:omega_def}), the resulting operator-valued expressions coincide for all $\omega$.
\end{definition}

\begin{definition}[Numerical equality]
	\label{dfn:numerical_equality}
	Let $e_1[\omega]$ and $e_2[\omega]$ be two real-valued functions corresponding to entries of a moment matrix. We say that they are \emph{numerically equal}, and write
	\begin{equation}
		e_1[\omega] = e_2[\omega],
	\end{equation}
	if their values coincide for a given realization of the random parameter $\omega$.
\end{definition}

Algebraic equality implies representational equality, which in turn implies numerical equality for all $\omega$. The distinction between algebraic, representational, and numerical equality is essential for the analysis carried out in this work. Algebraic equality ($\algeq$) captures identities that follow solely from the intrinsic properties of projectors and therefore hold for all quantum realizations. Representational equality ($\equiv$) refers to identities that arise after substituting a specific parametrization of the operators and the state, and thus may depend on the structure of this parametrization (e.g.\ in the rank-$1$ case). Finally, numerical equality ($=$) concerns the coincidence of values for a particular realization of the random parameter $\omega$. In general, algebraic equality implies representational equality, which in turn implies numerical equality for all $\omega$, while the converse implications do not hold. This hierarchy is crucial in our approach: the NPA constraints correspond to algebraic equalities, while our method detects them through numerical equalities obtained from random sampling, relying on the fact that non-algebraic equalities occur only with probability $0$ under generic conditions.

\section{Algebraic equality for entries of the moment matrix}
\label{app:algebraic}

In the following lemma we state a condition under which two entries of a moment matrix are equal for algebraic reasons, i.e., due to identities between the corresponding operator products, as specified in Definition~\ref{dfn:algebraic_equality}.

\begin{lem}
	\label{lem:moment_entry_equality}
	Let $\rho$ be a quantum state on a Hilbert space $\mathcal{H} = \bigotimes_{i=0}^{m-1} \mathcal{H}^{(i)}$, and let $\mathbb{P}^{(i)}$ and $\mathbb{P}'^{(i)}$ be simplified blocks of projectors acting on $\mathcal{H}^{(i)}$ for each $i \in \{0, \dots, m-1\}$. Denote by $\hat{\mathbb{P}}^{(i)}$ and $\hat{\mathbb{P}}'^{(i)}$ the corresponding scenario operators as in~\eqref{eq:scenario_operator}. Define the moment matrix entries
	\begin{equation}
		e_1 := \Tr\left( \rho \, \hat{\mathbb{P}}^{(0)} \otimes \dots \otimes \hat{\mathbb{P}}^{(m-1)} \right),
	\end{equation}
	\begin{equation}
		e_2 := \Tr\left( \rho \, \hat{\mathbb{P}}'^{(0)} \otimes \dots \otimes \hat{\mathbb{P}}'^{(m-1)} \right).
	\end{equation}
	Then
	\begin{equation}
		e_1 \algeq e_2 
		\quad \iff \quad 
		\hat{\mathbb{P}}^{(i)} \algeq \hat{\mathbb{P}}'^{(i)} 
		\quad \bigforall i \in \{0, \dots, m-1\}.
	\end{equation}
\end{lem}

Here $\algeq$ denotes equality of operators induced by the algebraic relations between projectors, i.e., orthogonality, idempotency, and commutation between different parties.

This result follows directly from Definitions~\ref{dfn:seq_of_projs} and~\ref{dfn:block_of_projs}. Since each block $\mathbb{P}^{(i)}$ is simplified, its associated operator $\hat{\mathbb{P}}^{(i)}$ is expressed in a reduced form with respect to these relations. Therefore, two tensor products of such operators are algebraically equal if and only if all corresponding local operators are algebraically equal.

\section{Random generation}
\label{app:random_generation}

In this section, we study the entries of a generic moment matrix, each of which takes the form:
\begin{equation}
	\label{eq:moment_generic}
	\Tr\left( \rho \hat{\mathbb{P}}^{(1)} \otimes \dots \otimes \hat{\mathbb{P}}^{(m)} \right),
\end{equation}
where $\rho$ is a quantum state and $\mathbb{P}^{(i)}$ is a simplified block of projectors associated with the measurements of the $i$-th party, with $\hat{\mathbb{P}}^{(i)}$ denoting the corresponding scenario operator as in~\eqref{eq:scenario_operator}. Recall that a block is \emph{simplified} if no further reduction in its length is possible using the algebraic properties of projectors (orthogonality and idempotency), as stated in Definition~\ref{dfn:block_of_projs}.

Let us consider two arbitrary entries of the moment matrix,
\begin{subequations}
	\begin{equation}
		e_1 := \Tr\left( \rho \hat{\mathbb{P}}^{(1)} \otimes \dots \otimes \hat{\mathbb{P}}^{(m)} \right),
	\end{equation}
	\begin{equation}
		e_2 := \Tr\left( \rho \hat{\mathbb{P}}'^{(1)} \otimes \dots \otimes \hat{\mathbb{P}}'^{(m)} \right),
	\end{equation}
\end{subequations}
and examine the conditions under which $e_1 = e_2$ holds when the quantum state and measurements are sampled at random, for fixed rank $r$ and Hilbert space dimension $d$. We focus here on the single-party case ($m = 1$), with the multi-party case being analogous.

We model both the quantum state and the measurements appearing in these expressions as measurable functions of a common random variable $\omega$ defined on an underlying probability space. Specifically, for $\ell_1, \ell_2 \in \mathbb{N}_{+}$, we define $\omega$ as the collection
\begin{equation}
	\label{eq:omega_def}
	\omega := 
	\left( p_{i,k,l} \right)_{\substack{i \in [\ell_1] \\ k \in [r] \\ l \in [d]}}
	\;\frown\;
	\left( p'_{j,k,l} \right)_{\substack{j \in [\ell_2] \\ k \in [r] \\ l \in [d]}}
	\;\frown\;
	\left( r_{k,l} \right)_{\substack{k \in [d] \\ l \in [d]}},
\end{equation}
where the symbol ``$\frown$'' denotes concatenation of sequences: for two sequences $a = (a_0, \dots, a_{n_1})$ and $b = (b_0, \dots, b_{n_2})$, we define $a \frown b := (a_0, \dots, a_{n_1}, b_0, \dots, b_{n_2})$.

We assume that the entries of $\omega$ are independent random variables with absolutely continuous distributions.

The variables in Eq.~\eqref{eq:omega_def} are connected to the state and measurement operators as follows:
\begin{subequations}
	\label{eqs:construction_rho_meas_from_omega}
	\begin{align}
		\label{eq:blockP}
		\mathbb{P} &= \big( P_0, P_1, \dots, P_{\ell_1 -1} \big), \\
		\label{eq:blockPprime}
		\mathbb{P}' &= \big( P'_0, P'_1, \dots, P'_{\ell_2 -1} \big),
	\end{align}
	with each individual rank-$r$ projector given by
	\begin{align}
		\label{eq:projP}
		P_i &= \sum_{k=0}^{r-1} \proj{p_{i,k}}, \\
		\label{eq:projPprime}
		P'_j &= \sum_{k=0}^{r-1} \proj{p'_{j,k}},
	\end{align}
	where the vectors are
	\begin{align}
		\ket{p_{i,k}} &:= 
		\begin{bmatrix}
			p_{i,k,0} \\
			p_{i,k,1} \\
			\vdots \\
			p_{i,k,d-1}
		\end{bmatrix}, &
		\ket{p'_{j,k}} &:= 
		\begin{bmatrix}
			p'_{j,k,0} \\
			p'_{j,k,1} \\
			\vdots \\
			p'_{j,k,d-1}
		\end{bmatrix}.
	\end{align}
	We assume that for each fixed $i$ (resp.\ $j$), the vectors $\{ \ket{p_{i,k}} \}_{k=0}^{r-1}$ (resp.\ $\{ \ket{p'_{j,k}} \}_{k=0}^{r-1}$) are orthonormal, so that $P_i$ and $P'_j$ are rank-$r$ projectors.
	
	The associated scenario operators are then
	\begin{equation}
		\label{eq:hat_mathbbP}
		\hat{\mathbb{P}} = P_0 P_1 \dots P_{\ell_1 -1},
	\end{equation}
	and
	\begin{equation}
		\label{eq:hat_mathbbPp}
		\hat{\mathbb{P}}' = P'_0 P'_1 \dots P'_{\ell_2 -1}.
	\end{equation}
The quantum state is parametrized as
	\begin{equation}
		\label{eq:rho}
		\rho := \sum_{k=0}^{d-1} c_k \ket{r_k} \bra{r_k},
	\end{equation}
	where each $\ket{r_k}$ is a vector in $\mathbb{C}^d$,
	\begin{equation}
		\label{eq:vec_r}
		\ket{r_k} :=
		\begin{bmatrix}
			r_{k,0} \\
			r_{k,1} \\
			\vdots \\
			r_{k,d-1}
		\end{bmatrix}.
	\end{equation}
\end{subequations}

This formal construction allows us to view entries in the moment matrix as functions of the random variable $\omega$, enabling a probabilistic analysis of when two such entries are equal under random sampling of the projectors and quantum state.

For the multi-party case $m > 1$, each party $i \in \{1, \dots, m\}$ is associated with its own simplified block of projectors $\mathbb{P}^{(i)} = \big( P^{(i)}_0, \dots, P^{(i)}_{\ell_i -1} \big)$, with the corresponding scenario operator:
\begin{equation}
	\hat{\mathbb{P}}^{(i)} := P^{(i)}_0 P^{(i)}_1 \dots P^{(i)}_{\ell_i -1}.
\end{equation}
An entry of the moment matrix then takes the form
\begin{equation}
	e = \Tr\Big( \rho \hat{\mathbb{P}}^{(1)} \otimes \hat{\mathbb{P}}^{(2)} \otimes \dots \otimes \hat{\mathbb{P}}^{(m)} \Big),
\end{equation}
where $\rho$ is a quantum state in the joint Hilbert space $\mathcal{H} = \bigotimes_{i=1}^m \mathcal{H}^{(i)}$. Random generation proceeds analogously by sampling, for each party independently, the vectors defining the projectors $P^{(i)}_j$, and sampling the global state $\rho$ on the entire tensor product space. The random variable $\omega$ then concatenates all such parameters for all parties.

Hence, equality of two multi-party moment matrix entries for all realizations (i.e.\ algebraically) reduces to equality of all local blocks $\mathbb{P}^{(i)}$, and thus of the corresponding operators $\hat{\mathbb{P}}^{(i)}$, across all parties, in direct analogy with the single-party case.

\section{Randomization results}
\label{app:randomization_results}

\begin{table}[h]
	\centering
	\begin{tabular}{|c|*{19}{c|}}
		\hline
		Length & 3 & 3 & 4 & 4 & 4 & 5 & 5 & 5 & 5 & 5 & 5 & 6 & 6 & 6 & 6 & 6 & 6 & 7 & 7 \\ \hline
		Rank & 3 & 2 & 3 & 4 & 3 & 2 & 3 & 3 & 2 & 3 & 4 & 2 & 3 & 2 & 2 & 3 & 3 & 3 & 2 \\ \hline
		Dimension & 6 & 9 & 15 & 8 & 15 & 6 & 11 & 18 & 10 & 9 & 24 & 14 & 12 & 11 & 6 & 6 & 12 & 12 & 6 \\ \hline
		Scenario & [2,2] & [4,4] & [2,4] & [2,2,2] & [5,2] & [2,3,3] & [2,3,3] & [4,4] & [2,2,5] & [3,3,3] & [2,6] & [7,2] & [3,3,3] & [2,2] & [3,2] & [2,2,2] & [2,3,4] & [4,4] & [3,2,3] \\
		\hline
	\end{tabular}
	\caption{Tested instances supporting Conjecture~\ref{conjecture1}. In all cases the rank $r$ of the generated projectors satisfies $r > 1$.}
	\label{tab:conjecture_1}
\end{table}

\begin{table}[h]
	\centering
	\begin{tabular}{|c|*{19}{c|}}
		\hline
		Length & 3 & 3 & 4 & 4 & 4 & 5 & 5 & 5 & 5 & 5 & 5 & 6 & 6 & 6 & 6 & 6 & 6 & 7 & 7 \\ \hline
		Dimension & 10 & 4 & 4 & 12 & 5 & 3 & 35 & 4 & 5 & 3 & 6 & 12 & 5 & 2 & 9 & 8 & 6 & 4 & 7 \\ \hline
		Scenario & [2,2] & [4,4] & [2,4] & [2,2,2] & [5,2] & [2,3,3] & [2,3,3] & [4,4] & [2,2,5] & [3,3,3] & [2,6] & [7,2] & [3,3,3] & [2,2] & [3,2] & [2,2,2] & [2,3,4] & [4,4] & [3,2,3] \\
		\hline
	\end{tabular}
	\caption{Tested instances supporting Conjecture~\ref{conjecture1} for rank-$1$ projectors. Only non-homogeneous pairs of blocks $\mathbb{P}$ and $\mathbb{P}'$ were considered.}
	\label{tab:conjecture_2}
\end{table}

To support Conjecture~\ref{conjecture1} below, we performed extensive numerical tests over a wide range of scenarios, block lengths, and Hilbert space dimensions, the details of which are summarized in Tables~\ref{tab:conjecture_1} and~\ref{tab:conjecture_2}. The tested instances were selected to cover a broad range of scenarios, block lengths, and Hilbert space dimensions.

\begin{conj}
	\label{conjecture1}
	Consider the following entries of a moment matrix associated with a single-party scenario:
	\begin{equation}
		e_1 := \Tr\big( \rho \hat{\mathbb{P}} \big),
	\end{equation}
	\begin{equation}
		e_2 := \Tr\big( \rho \hat{\mathbb{P}}' \big),
	\end{equation}
	with $e_1 \notalgeq e_2$ (Lemma~\ref{lem:moment_entry_equality}), where $\hat{\mathbb{P}}$ and $\hat{\mathbb{P}}'$ are operators constructed from the single-agent blocks of projectors $\mathbb{P}$ and $\mathbb{P}'$, respectively, associated with two entries of the moment matrix.

	Then, for any finite Hilbert space dimension $d$, under the random generation model described in Section~\ref{app:random_generation}, the probability that a randomly generated parameter configuration $\omega$ satisfies
	\begin{equation}
		\label{eq:lem_spp_e1e2}
		e_1[\omega] = e_2[\omega]
	\end{equation}
	is equal to $0$ if either
	\begin{equation}
		r > 1
		\quad \text{or} \quad
		\mathbb{P} \text{ and } \mathbb{P}' \text{ are not homogeneous}.
	\end{equation}
\end{conj}

\begin{proof}
	The statement is equivalent to showing that the following system is satisfied with probability $0$ for a randomly generated $\omega$:
	\begin{equation}
		\label{eq:systemP3}
		\left\{
		\begin{array}{ll}
			e_1[\omega] = e_2[\omega], \\
			r>1 \text{ or } \hat{\mathbb{P}}, \hat{\mathbb{P}}' \text{ are \emph{not} homogeneous}.
		\end{array}
		\right.
	\end{equation}

	By construction (see Section~\ref{app:random_generation}), both $e_1[\omega]$ and $e_2[\omega]$ are obtained by substituting the random parameters $\omega$ into the expressions
	\begin{equation}
		e_1 = \Tr(\rho \hat{\mathbb{P}}), \quad e_2 = \Tr(\rho \hat{\mathbb{P}}').
	\end{equation}
	Hence, their difference
	\begin{equation}
		f(\omega) := e_1[\omega] - e_2[\omega]
	\end{equation}
	is a function of the random parameters $\omega$.

	In the idealized setting, $f(\omega)$ is expected to have a polynomial (or, more generally, analytic) dependence on $\omega$. Under this assumption, either $f(\omega)$ is identically zero, or the set of solutions of the equation
	\begin{equation}
		f(\omega) = 0
	\end{equation}
	has measure zero.

	Therefore, to distinguish between these two cases, it is sufficient to test the equality numerically for randomly generated instances. If $f(\omega) \neq 0$ for a randomly sampled $\omega$, then this is consistent with the hypothesis that the solution set has measure zero.

	In our implementation, instead of sampling $\omega$ directly, we generate random states and measurements according to Algorithms~\ref{alg:HaarDensity} and~\ref{alg:random_rank_r_projectors_2}. Although we do not prove that this induces a strictly polynomial dependence on the underlying random parameters, we conjecture that the resulting distribution is sufficiently generic to preserve the measure-zero property of non-trivial solutions.

	In Tables~\ref{tab:conjecture_1} and~\ref{tab:conjecture_2} we report numerical results for systems of the form~\eqref{eq:systemP3}, for different values of the rank of the projectors, the considered scenarios, the dimension of the quantum state, and the length of the blocks. For each case, we generated $500{,}000$ instances and checked whether $e_1[\omega] = e_2[\omega]$ holds up to a tolerance of $10^{-16}$. In all tested cases, the equality was never satisfied, providing strong numerical evidence supporting the conjecture.
\end{proof}

This is consistent with the general principle that zeros of non-trivial analytic functions form sets of measure zero.

We now prove the following Lemma~\ref{lem:single_party_probability}.

\begin{lem}
	\label{lem:single_party_probability}
	Consider the following entries of a moment matrix associated with a single-party scenario:
	\begin{equation}
		e_1 := \Tr\big( \rho \hat{\mathbb{P}} \big),
	\end{equation}
	\begin{equation}
		e_2 := \Tr\big( \rho \hat{\mathbb{P}}' \big),
	\end{equation}
	with $e_1 \notalgeq e_2$ (Lemma~\ref{lem:moment_entry_equality}), where $\hat{\mathbb{P}}$ and $\hat{\mathbb{P}}'$ are operators constructed from the single blocks of projectors $\mathbb{P}$ and $\mathbb{P}'$, respectively.
	
	Then, assuming Conjecture~\ref{conjecture1}, for any finite Hilbert space dimension $d$, under the random generation model described in Section~\ref{app:random_generation}, the probability that a randomly generated parameter configuration $\omega$ satisfies
	\begin{equation}
		e_1[\omega] = e_2[\omega]
	\end{equation}
	is equal to $0$ if and only if 
	\begin{equation}
		r > 1 
		\quad \text{or} \quad 
		\hat{\mathbb{P}} \text{ and } \hat{\mathbb{P}}' \text{ are not homogeneous},
	\end{equation}
	in the sense of Definition~\ref{dfn:homogeneous_blocks}. Otherwise, the probability is $1$.
\end{lem}

\begin{proof}
	We denote by $\Prob(e_1[\omega] = e_2[\omega])$ the probability that a randomly generated parameter configuration $\omega$ satisfies $e_1[\omega] = e_2[\omega]$.
		
	The direction ``$\impliedby$'' is a direct consequence of Conjecture~\ref{conjecture1}.
	
	To prove ``$\implies$'', we assume that $r = 1$ and that $\hat{\mathbb{P}}$ and $\hat{\mathbb{P}}'$ are homogeneous, and show that $\Prob(e_1[\omega] = e_2[\omega]) = 1$. Consider the equality
	\begin{equation}
		\label{eq:thesisLemma3}
		\Tr \big( \rho \hat{\mathbb{P}} \big)\big|_{\omega} = 
		\Tr \big( \rho \hat{\mathbb{P}}' \big)\big|_{\omega}.
	\end{equation}
	
	Let us substitute the explicit expressions for the case $r = 1$. Using Eqs.~\eqref{eq:hat_mathbbP},~\eqref{eq:hat_mathbbPp},~\eqref{eq:projP},~\eqref{eq:projPprime} and~\eqref{eq:rho}, the equality~\eqref{eq:thesisLemma3} becomes
	\begin{equation}
		\label{eq:thesisLemma3_after_substitution}
		\sum_{k=0}^{d-1} c_k
		\bra{r_k} p_0 \rangle
		\bra{p_0} p_1 \rangle
		\bra{p_1} p_2 \rangle \cdots
		\bra{p_{\ell-1}} r_k \rangle
		=
		\sum_{k=0}^{d-1} c_k
		\bra{r_k} p'_0 \rangle
		\bra{p'_0} p'_1 \rangle
		\bra{p'_1} p'_2 \rangle \cdots
		\bra{p'_{\ell-1}} r_k \rangle.
	\end{equation}
	
	Since $\mathbb{P}$ and $\mathbb{P}'$ are \emph{homogeneous}, they have the same first and last projectors, hence
	\begin{equation}
		\bra{r_k} p_0 \rangle = \bra{r_k} p'_0 \rangle \quad \bigforall k,
	\end{equation}
	and
	\begin{equation}
		\bra{p_{\ell-1}} r_k \rangle = \bra{p'_{\ell-1}} r_k \rangle \quad \bigforall k.
	\end{equation}
	
	Moreover, homogeneity implies that $\mathbb{P}$ and $\mathbb{P}'$ have the same multiset of consecutive pairs. Therefore, the products of overlaps coincide:
	\begin{equation}
		\bra{p_0} p_1 \rangle
		\bra{p_1} p_2 \rangle \cdots
		\bra{p_{\ell-2}} p_{\ell-1} \rangle
		=
		\bra{p'_0} p'_1 \rangle
		\bra{p'_1} p'_2 \rangle \cdots
		\bra{p'_{\ell-2}} p'_{\ell-1} \rangle.
	\end{equation}
	
	Combining the above, we obtain $e_1[\omega] = e_2[\omega]$ for all $\omega$. Hence, $\Prob(e_1[\omega] = e_2[\omega]) = 1$.
\end{proof}

This shows that in the rank-$1$ homogeneous case, the equality is structural rather than probabilistic.

We now prove Lemma~\ref{conj:probability_equality} based on the single-party analysis.

\section{Main results}
\label{app:appendixLemma}

In the following lemma we formulate the condition under which two entries of a moment matrix, obtained from randomly generated projectors and quantum states, are equal, and we identify when additional (non-algebraic) equalities may arise with respect to Lemma~\ref{lem:moment_entry_equality}.

\begin{lem}
	\label{conj:probability_equality}
	Let us consider the following two entries of a moment matrix:
	\begin{equation}
		e_1 := \Tr\Big( \rho \hat{\mathbb{P}}^{(0)} \otimes \dots \otimes \hat{\mathbb{P}}^{(m-1)} \Big),
	\end{equation}
	\begin{equation}
		 e_2 := \Tr\Big( \rho \hat{\mathbb{P}}'^{(0)} \otimes \dots \otimes \hat{\mathbb{P}}'^{(m-1)} \Big),
	\end{equation}
	with $e_1 \notalgeq e_2$ (Lemma~\ref{lem:moment_entry_equality}), where each $\hat{\mathbb{P}}^{(i)}$ and $\hat{\mathbb{P}}'^{(i)}$ is the scenario operator associated with simplified blocks $\mathbb{P}^{(i)}$ and $\mathbb{P}'^{(i)}$, respectively. 
	
	Assuming Conjecture~\ref{conjecture1}, for any finite Hilbert space dimension $d$, and under the random generation model described in Section~\ref{app:random_generation}, the probability that a randomly generated parameter configuration $\omega$ satisfies 
	\begin{equation}
		 e_1[\omega] = e_2[\omega]
	\end{equation}
	is equal to $0$ if 
	\begin{equation}
		r > 1 \quad \text{or} \quad \exists i \in [m] \ \text{such that } \mathbb{P}^{(i)} \text{ and } \mathbb{P}'^{(i)} \text{ are not homogeneous},
	\end{equation}
	in the sense of Definition~\ref{dfn:homogeneous_blocks}.
	
	In the remaining cases (i.e.\ when $r = 1$ and all corresponding blocks are homogeneous), the probability of equality is $1$. This corresponds to the fact that in these cases the equality $e_1[\omega] = e_2[\omega]$ holds identically for all realizations generated by $\omega$.
\end{lem}

Here, the notation $e_1[\omega]$ denotes the evaluation of the entry $e_1$ after substituting the randomly generated parameters $\omega$ that define both the state $\rho$ and the projectors comprising $\hat{\mathbb{P}}^{(i)}$, and similarly for $e_2[\omega]$.

Without loss of generality, we can restrict our attention to the single-party case, where $m = 1$ and we denote 
\begin{equation}
	\mathbb{P} := \mathbb{P}^{(0)}, \quad 
	\mathbb{P}' := \mathbb{P}'^{(0)}, \quad
	\hat{\mathbb{P}} := \hat{\mathbb{P}}^{(0)}, \quad
	\hat{\mathbb{P}}' := \hat{\mathbb{P}}'^{(0)}.
\end{equation}

To see why it is sufficient to consider a single party, let us illustrate this with a two-party case. Suppose we have the equality:
\begin{equation}
	\label{eq:two_party_case}
	\Tr\Big( \rho \hat{\mathbb{P}}^{(1)} \otimes \hat{\mathbb{P}}^{(2)} \Big)
	=
	\Tr\Big( \rho \hat{\mathbb{P}}'^{(1)} \otimes \hat{\mathbb{P}}'^{(2)} \Big),
\end{equation}
where $\hat{\mathbb{P}}^{(1)}, \hat{\mathbb{P}}^{(2)}, \hat{\mathbb{P}}'^{(1)}, \hat{\mathbb{P}}'^{(2)}$ are scenario operators built from simplified blocks of projectors. 

For Eq.~\eqref{eq:two_party_case} to hold for all quantum states $\rho$, it is necessary that 
\begin{equation}
	\hat{\mathbb{P}}^{(1)} \equiv \hat{\mathbb{P}}'^{(1)}
	\quad \text{and} \quad
	\hat{\mathbb{P}}^{(2)} \equiv \hat{\mathbb{P}}'^{(2)}.
\end{equation}
Here $\equiv$ denotes equality of the corresponding operator representations, i.e.\ equality after substituting the parametrization induced by $\omega$. In particular, algebraic equality ($\algeq$) implies representational equality ($\equiv$), but not conversely.

If at least one of these conditions is violated, then the equality in Eq.~\eqref{eq:two_party_case} imposes a non-trivial constraint on $\rho$ (i.e.\ one that is not satisfied for all parameter values in Eq.~\eqref{eq:rho}). Consequently, the set of $\rho$ satisfying this equality has measure zero under the random generation model.

Since the same reasoning applies independently for each party $i \in [m]$, it is therefore sufficient to analyze the single-party scenario to determine the conditions under which equality between entries of the moment matrix holds.

\begin{proof}[Proof of Lemma~\ref{conj:probability_equality}]
	As shown above, it is sufficient to consider the single-party case. Indeed, for a multi-party scenario, equality of two moment matrix entries,
	\begin{equation}
		e_1[\omega] = e_2[\omega],
	\end{equation}
	requires simultaneous equality of all local contributions corresponding to each party.

	Therefore, if there exists at least one party $i \in [m]$ such that the corresponding local operators satisfy
	\begin{equation}
		\hat{\mathbb{P}}^{(i)} \not\equiv \hat{\mathbb{P}}'^{(i)},
	\end{equation}
	then the equality imposes a non-trivial constraint on the parameters associated with that party, and thus occurs with probability $0$ under random generation.

	Consequently, the probability that $e_1[\omega] = e_2[\omega]$ is equal to $1$ if and only if, for all parties $i \in [m]$, the corresponding single-party equalities hold with probability $1$.

	Applying Lemma~\ref{lem:single_party_probability}, this happens if and only if for all parties $i$ we have
	\begin{equation}
		r = 1 \quad \text{and} \quad \mathbb{P}^{(i)} \text{ and } \mathbb{P}'^{(i)} \text{ are homogeneous}.
	\end{equation}

	Equivalently, if $r > 1$ or there exists at least one party for which the corresponding blocks are not homogeneous, then the equality holds with probability $0$.

	This proves the statement of Lemma~\ref{conj:probability_equality}.
\end{proof}

\bibliographystyle{ieeetr}
\bibliography{thetool_refs}

\end{document}